\documentclass[final, conference, twocolumn]{IEEEtran}
\usepackage{amsthm,amsmath,etex,amssymb,latexsym,comment,etoolbox,cancel,pdfpages,graphicx,tikz,pgfplots,amsbsy,multicol,color,xcolor,adforn,enumerate,framed,relsize}
\usepackage{mathtools}
\mathtoolsset{showonlyrefs}
\usepackage[T1]{fontenc}
\usepackage{mathrsfs}
\usepackage{ifthen}
\usepackage{qtxmath}
\usepackage{bbm}
\usepackage[noadjust]{cite}
\usepackage[all,cmtip,rotate]{xy}
\usepackage[framemethod=tikz]{mdframed}
\newboolean{islongversion}   
\setboolean{islongversion}{false}   
\usetikzlibrary{positioning,chains,fit,shapes,calc,patterns}
\definecolor{myblue}{RGB}{20,100,160}
\definecolor{mygreen}{RGB}{80,160,80}
\definecolor{myorange}{RGB}{130,80,50}
\usepackage[T1]{fontenc}
\newcommand{\szero}{\emptyset}
\newboolean{arxiv}
\setboolean{arxiv}{true}
\ifthenelse{\boolean{arxiv}}{\def\app{appendix}}{\def\app{extended manuscript~\cite{BakK:23a}}}

\ifthenelse{\boolean{arxiv}}{\def\applem#1{Lemma~\ref{#1}}}{\def\applem#1{\cite[Lemma~\ref{#1}]{BakK:23a}}}
\ifthenelse{\boolean{arxiv}}{\def\appeq#1{Eq.~\ref{#1}}}{\def\appeq#1{\cite[Eq.~\ref{#1}]{BakK:23a}}}
\ifthenelse{\boolean{arxiv}}{}{}
\ifthenelse{\boolean{arxiv}}{\def\appdef#1{Definition~\ref{#1}}}{\def\appdef#1{\cite[Definition~\ref{#1}]{BakK:23a}}}
\ifthenelse{\boolean{arxiv}}{
\usepackage{setspace}
\onehalfspacing\onecolumn\newcommand\figscale{2.5}}{\newcommand\figscale{1}}

\usepackage{url}

\allowbreak
\newboolean{isdraft}   
\newlength{\figurewidth}
\makeatletter
\if@twocolumn
\else
\fi
\makeatother

\makeatletter
\if@twocolumn\newcommand{\eqbr}[1][]{\\&\  \ifx\relax#1\relax\else #1\fi}
\else\newcommand{\eqbr}[1][]{}
\fi
\makeatother

\makeatletter
\if@twocolumn
\else
\fi
\makeatother

\makeatletter
\if@twocolumn
\else
\fi
\makeatother

\makeatletter
\if@twocolumn
\else
\fi
\makeatother

\ifthenelse{\boolean{isdraft}}
	{
		\newcommand{\longonly}[1]{\ifmmode #1\else {\color{red}#1}\fi}
		\newcommand{\shortonly}[1]{\ifmmode #1\else {\color{blue}#1}\fi}
		\newcommand{\shortlong}[2]{\ifmmode #1 #2\else {\color{red}#1} {\color{blue}{#2}}\fi}
	}
	{
		\ifthenelse{\boolean{islongversion}}
			{
				\newcommand{\longonly}[1]{#1}
				\newcommand{\shortonly}[1]{}
				\newcommand{\shortlong}[2]{#2}
			}
			{
				\newcommand{\longonly}[1]{}
				\newcommand{\shortonly}[1]{#1}
				\newcommand{\shortlong}[2]{#1}
			}
	}

\allowdisplaybreaks

\newcommand{\suppress}[1]{}

\theoremstyle{definition}

\providecommand{\corollaryname}{Corollary}
\providecommand{\theoremname}{Theorem}
\newtheorem{definition}{Definition}

\newtheorem{example}{Example}

\theoremstyle{plain}
\newtheorem{theorem}{Theorem}
\newtheorem{corollary}{Corollary}
\newtheorem{lemma}{Lemma}

\theoremstyle{remark}
\newtheorem{remark}{Remark}

\DeclareSymbolFont{lettersA}{U}{txmia}{m}{it}
 \DeclareMathSymbol{\bbr}{\mathbb}{lettersA}{"92}
 \DeclareMathSymbol{\bbc}{\mathord}{lettersA}{"83}
 \DeclareMathSymbol{\Exp}{\mathord}{lettersA}{"85}
 \DeclareMathSymbol{\ent}{\mathord}{lettersA}{"88}
 \DeclareMathSymbol{\MI}{\mathord}{lettersA}{"89}
 \DeclareMathSymbol{\KL}{\mathord}{lettersA}{"84}
\DeclareMathSymbol{\EL}{\mathord}{lettersA}{"8C}
\DeclareMathSymbol{\EM}{\mathord}{lettersA}{"8D}

  \DeclareMathSymbol{\bbf}{\mathord}{lettersA}{"86}
   \DeclareMathSymbol{\bbn}{\mathbb}{lettersA}{"8E}
 \DeclareMathSymbol{\bbg}{\mathord}{lettersA}{"87}
 
\makeatletter
\renewenvironment{proof}[1][\proofname]{\par
  \pushQED{\qed}%
  \normalfont \topsep6\p@\@plus6\p@\relax
  \trivlist
  \item[\hskip\labelsep
        \itshape
    #1\@addpunct{:}]\mbox{}\\*
}{%
  \popQED\endtrivlist\@endpefalse
}
\makeatother

\newcommand{\cA}{\mathscr{A}}
\newcommand{\cB}{\mathscr{B}}

\newcommand{\cP}{\mathscr{P}}

\newcommand{\cS}{\mathscr{S}}

\newcommand{\cX}{\mathscr{X}}
\newcommand{\cY}{\mathscr{Y}}
\newcommand{\cZ}{\mathscr{Z}}

\newcommand{\Prob}{P}
\newcommand{\Probc}{Q}

\newcommand{\wy}{\ensuremath{{{W}}_{Y|X,S}}}
\newcommand{\uz}{\ensuremath{{U}_{Z|X}}}

\newcommand{\authenc}{P_{X^n|M}}
\newcommand{\authdec}{\phi}
\newcommand{\accept}{\mathsf{ACCEPT}}
\newcommand{\reject}{\mathsf{REJECT}}
\newcommand{\fP}{\Gamma}
\newcommand{\distribution}{\uz-distribution }

\newcommand{\cst}{C_{\textrm{stoch,auth}}}

\newcommand{\set}[1]{\left\{#1\right\}}
\newcommand{\expect}[2]{{\normalfont\textsf{\bf E}}_{#1}\left[#2\right]}
\usepackage{authblk}
\newcommand{\Ctag}{C_{\mathrm{tag}}}

\newcommand{\Perr}{P_{\mathrm{e}}}

\IEEEoverridecommandlockouts \title{On Authentication against a Myopic Adversary using Stochastic Codes\thanks{This material is based upon work supported by the National Science Foundation under Grant No. CCF-1908725 and CCF-2107526.
 Author emails: mayank.bakshi@ieee.org, okosut@asu.edu.}}
\author{Mayank Bakshi}
\author{Oliver Kosut}
\affil{Arizona State University}

\begin{document}
\bstctlcite{IEEEexample:BSTcontrol}
\maketitle
\begin{abstract} We consider the problem of authenticated communication over a discrete arbitrarily varying channel where the legitimate parties are unaware of whether or not an adversary is present. When there is no adversary, the channel state always takes a default value $\szero$. When the adversary is present, they may choose the channel state sequence based on a non-causal noisy view of the transmitted codewords and the encoding and decoding scheme. We require that the  decoder output the correct message with a high probability when there is no adversary, and either output the correct message or reject the transmission when the adversary is present. Further, we allow the transmitter to employ private randomness during encoding that is known neither to the receiver nor the adversary. Our first result proves a dichotomy property for the capacity for this problem -- the capacity either equals zero or it equals the non-adversarial capacity of the channel. Next, we give a sufficient condition for the capacity for this problem to be positive even when  the non-adversarial channel to the receiver is stochastically degraded with respect to the channel to the adversary. Our proofs rely on a connection to a standalone authentication problem, where the goal is to accept or reject a candidate message that is already available to the decoder. Finally, we give examples and compare our sufficient condition with other related conditions known in the literature.\end{abstract}
\section{Introduction}

Consider the problem of communication over a channel where an adversary may or may not be present.  Neither the transmitter nor the receiver know {\em a priori} whether or not the adversary is present. The goal for the transmission is to decode the message with an authentication guarantee. In particular, when the adversary is not present, it is desirable that the message is decoded correctly with a high probability. On the other hand, when the adversary is present, the decoding goal is relaxed -- the decoder may either output the correct message, or it may declare that the adversary is present.

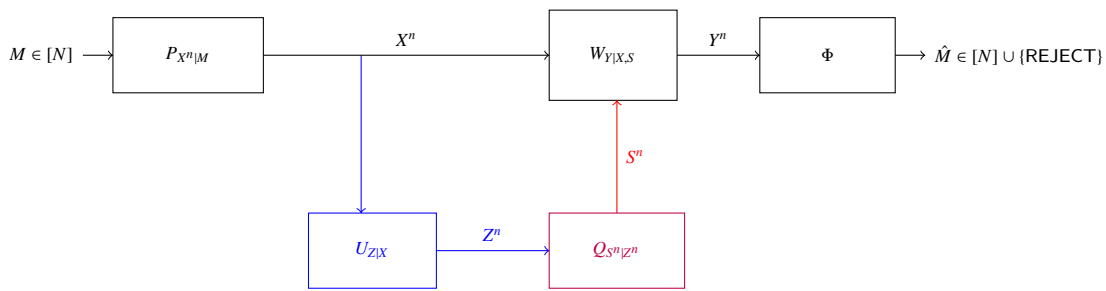
\begin{figure}
 	\centering
	\begin{tikzpicture}[scale=\figscale*0.4]
	\draw[color=white] (-0.2-2.6,5.4) rectangle ++(2,1) node[pos=.5]{\scriptsize Transmitter};
	\draw (-0.2-2.6,6.5) rectangle ++(2,1) node[pos=.5]{\scriptsize $\Prob_{X^n|M}$} ;
	\draw[<-,color=blue] (0.5, 4.9) -- ++ (0, 2.1); 
	\draw[color=blue] (0.5-0.9+0.2,3.9) rectangle ++(2.5-0.2-0.6,1) node[pos=.5]{\scriptsize\color{blue} $\uz$};
	\draw[->,color=blue] (1.5, 4.4) -- node[above]{\scriptsize \color{blue} $Z^n$} ++(1.5,0);
	
	\draw (3.2-0.6+0.4,6.4) rectangle ++(2.5-0.8,1.2) node[pos=.5]{\scriptsize $\wy$};
	\draw[->,color=red] (3.9, 4.9) -- node[right]{\scriptsize \color{red} $S^n$} ++ (0, 1.5) ;
	\draw[color=purple] (2.6+0.4, 3.9) rectangle ++(2.6-0.8,1) node[pos=0.5]{\scriptsize\color{purple} $\Probc_{S^n|Z^n}$};
	\draw[color=white] (2.6+0.4, 3.9-1.1) rectangle ++(2.6-0.8,1) node[pos=0.5]{\scriptsize Adversary};
	\draw[->] (-3.2,6.5+0.5) node[anchor=east]{\scriptsize $M\in[N]$} -- ++ (0.4,0) ;
	\draw[->] (2-0.2-2.6, 6.5+0.5) -- node[above]{\scriptsize $X^n$} ++(1.4+2+0.4,0);
	\draw[->] (5.1-0.4,6.5+0.5) -- node[above] {\scriptsize $Y^n$} ++ (1.7+0.4-1,0);
	\draw (6.8-1, 6+0.5) rectangle ++(2.6-0.8,1) node[pos=0.5]{\scriptsize $\Phi$};
	\draw[color=white] (6.8-1, 6+0.5-1.1) rectangle ++(2.6-0.8,1) node[pos=0.5]{\scriptsize Receiver};
	\draw[->] (9.4-0.8-1, 6.5+0.5) -- ++ (0.7-0.3,0)  node[right] {\scriptsize $\hat{M}\in[N]\cup\set{\reject}$};
	\node[align=center, color = red] at (7,3) {{\footnotesize }};
\end{tikzpicture}
\caption{The authentication problem over a myopic arbitrarily varying channel.
}\label{fig:authcomm}
\end{figure}

We study this problem in the setting of \emph{myopic arbitrarily varying channels (myopic AVCs)} (see Fig~\ref{fig:authcomm}). The channel between the transmitter and the receiver is an Arbitrarily Varying Channel (AVC) $\wy$. When the adversary is absent, the channel state assumes a default ``no-adversary'' state $\szero$. On the other hand, when the adversary is present, they may choose the state sequence arbitrarily over the entire transmission. The adversary is \emph{myopic}, \emph{i.e.}, they have a non-causal view of the codeword passed through a memoryless channel \uz\ \underline{before} they select the channel state sequence. The adversary's observation is conditionally independent of the receiver's observation given the transmitted codeword and the channel state sequence.

\subsection{Related work}

The reliable communication problem over AVCs has a rich history~\cite{Blachman62,Ahlswede78, CsiszarN88}. Myopic AVCs have been studied in~\cite{Sarwate:10,DeyJL:15,ZhangVJS:18,BudkuleyDJLSW:20}.
The problem of authentication has been studied in several different frameworks. There is a long line of work that examines the message authentication problem -- both as a standalone problem~\cite{SimmonsCRYPTO84, MaurerIT00} as well as over noisy channels~\cite{LaiEPIT09,TuLIT18}. In recent years, considerable attention has focused on keyless authentication over adversarial channels, which is the setting of this paper~\cite{KosutKITW18,GravesYS:16,GravesBKKY:23,SangwanBDP:19ISIT,BeemerGKKY:20a}. Authentication over myopic AVCs has previously been studied in~\cite{BeemerKKGY:19,BeemerGKKY:20}.

\subsection{Our contribution}
In the following, we summarize our main results. The proofs of these results are detailed in later sections after formally describing the problem and notation in Section~\ref{sec:setup}. 

In this paper, we consider the capacity $\cst$ for this problem when the encoders are allowed to be \emph{stochastic} of the form $P_{X^n|M}$, \emph{i.e.}, the transmitter may employ private randomness while encoding. This randomness is neither available to the receiver nor to the adversary. 

\subsubsection{Connection to Authentication Tags}
\begin{figure}
	
 	\centering
	\begin{tikzpicture}[scale=\figscale*0.4]
	\draw[white] (-0.2,3.5) rectangle ++(3.4,3.7) ;
	\draw (-0.2-2.6,6.5) rectangle ++(2,1) node[pos=.5]{\scriptsize $\Prob_{X^n|M}$};
	\draw[<-,color=blue] (0.5, 4.9) -- ++ (0, 2.1); 
	\draw[color=blue] (0.5-0.9+0.2,3.9) rectangle ++(2.5-0.2-0.6,1) node[pos=.5]{\scriptsize\color{blue} $\uz$};
	\draw[->,color=blue] (1.5, 4.4) -- node[above]{\scriptsize \color{blue} $Z^n$} ++(1.5,0);
	
	\draw (3.2-0.6+0.4,6.4) rectangle ++(2.5-0.8,1.2) node[pos=.5]{\scriptsize $\wy$};
	\draw[->,color=red] (3.9, 4.9) -- node[right]{\scriptsize \color{red} $S^n$} ++ (0, 1.5) ;
	\draw[color=purple] (2.6+0.4, 3.9) rectangle ++(2.6-0.8,1) node[pos=0.5]{\scriptsize\color{purple} $\Probc_{S^n|Z^n}$};
	\draw[->] (-3.2,6.5+0.5) node[anchor=east]{\scriptsize $M\in[N]$} -- ++ (0.4,0) ;
	\draw[->] (2-0.2-2.6, 6.5+0.5) -- node[above]{\scriptsize $X^n$} ++(1.4+2+0.4,0);
	\draw[->] (5.1-0.4,6.5+0.5) -- node[above] {\scriptsize $Y^n$} ++ (1.7+0.4-1,0);
	\draw (6.8-1, 6+0.5) rectangle ++(2.6-0.8,1) node[pos=0.5]{\scriptsize $\phi_{\hat{m}}$};
	\draw[->] (9.4-0.8-1, 6.5+0.5) -- ++ (0.7-0.3,0)  node[right] {\scriptsize $A\in\set{\accept,\reject}$};
\end{tikzpicture}
\caption{In the authentication tag problem, the objective is to authenticate a candidate message $\hat{m}$ that is already available to the decoder. When no adversary is present, it is desirable that the decoder outputs a $\reject$ if and only if the candidate message $\hat{m}$ is not equal to the true message $m$. When there is an adversary present, the decoder may output a $\reject$ to either denote that the $\hat{m}$ is not equal to $m$ or to indicate that an adversary has been detected.
}\label{fig:authtag}
\end{figure}
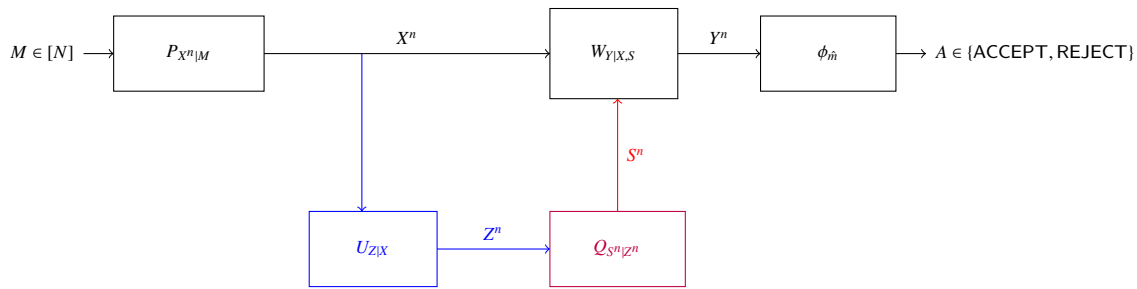

Consider the authentication tag problem shown in Figure~\ref{fig:authtag}.  This problem is reminiscent of the problem of \emph{identification via channels}~\cite{AhlswedeD:89}, and in a similar vein, supports the number of messages to be doubly exponential in the blocklength. Our first result draws a connection between the authentication problem and the {\em authentication tag} problem. We show that authentication tag capacity $C_{\mathrm{tag}}$ equals $\cst$ for all myopic AVCs.  Let $C_{\szero}$ denote the channel capacity (\emph{i.e.}, the non-adversarial capacity) of the channel $\mathsf{W}^{(\szero)}_{Y|X}(\cdot|\cdot)\triangleq \mathsf{W}_{Y|X,S}(\cdot|\cdot,\szero)$.  
\begin{theorem}\label{thm:dichotomy} $C_{\textrm{tag}}= C_{\textrm{stoch,auth}}$. Further, whenever these are positive, they equal $C_{\szero}$.
\end{theorem}
While this result has been previously noted in settings where the adversary is oblivious of the transmission (\emph{c.f.},~\cite{SangwanBDP:19ISIT}), to the best of our knowledge, this is the first extension of this property to myopic adversaries. Thus, Theorem~\ref{thm:dichotomy} shows that, from a capacity viewpoint, it is sufficient to examine authentication tags. \subsubsection{Overwritability condition} Theorem~\ref{thm:dichotomy} alludes to a dichotomy property for the authentication problem with myopic adversaries. Such dichotomies are well known in the AVC literature for both the reliable communication problem as well as the authentication problem. In the authentication setting, this dichotomy is often characterized via an appropriate overwritability criterion that specifies the condition under which the adversary can confuse the receiver between legitimate (non-adversarial) transmission a symbol $x'$ at the channel input and an adversarially influenced transmission. 

Coming to authentication over myopic AVCs,~\cite{BeemerKKGY:19} introduces the notion of $\uz$-overwritability (also see ~\appdef{def:uow}). As noted in~\cite{BeemerKKGY:19}, whenever a channel $\wy$ is \uz-overwritable, the authentication capacity equals zero. Further, if encoding is restricted to deterministic codes, the authentication capacity is zero whenever the channel $\wy$ is stochastically degraded with respect to \uz, \emph{i.e.}, there exists a channel $\mathsf{V}_{Y|Z}$ such that $\wy (y|x,\szero)=\sum_{z\in\cZ}\uz (z|x)\mathsf{V}_{Y|Z}(y|z) \forall y\in\cY,x\in\cX,$ 
and $\wy$ satisfies the $I$-overwritability condition (see~\appdef{def:iow}). Going beyond deterministic codes,~\cite{BeemerGKKY:20} gives an example to show that the authentication capacity may be non-zero even if the channel is $I$-overwritable and $\wy$ is stochastically degraded with respect to $\uz$. 

Our next result is motivated by this example to give a general sufficient condition for positivity of authentication capacity. We present a new overwritability condition that we call $\uz$-distribution overwritability. Intuitively, this condition requires that the myopic adversary be able to overwrite the channel output to mimic legitimate transmission of any symbol $x'$ of their choice when the true input to the channel is drawn from any distribution that leads to a publicly known distribution $P_Z$ for the adversary (see Definition~\ref{def:udistow} for a formal statement). 
\begin{theorem} $\cst>0$ if $\wy$ is not \distribution overwritable.	\label{thm:achievability}
\end{theorem}

On the way to proving positive rates for channels that are not $\uz$-distribution overwritable, we first prove achievability for a smaller class of channels that are not $(\uz,\cX)$-distribution overwritable (see~\appdef{def:upow}). Our characterizations of overwritability give rise a natural hierarchy among different notions of overwritability for myopic AVCs. Figure~\ref{fig:hierarchy} shows the inclusion relationships between the overwritability notions we touch upon in this paper. In Section~\ref{sec:examples}, we give examples to show that these inclusions are, in fact, strict. \ifthenelse{\boolean{arxiv}}{}{In order to keep the exposition here short, we present some of the technical details in the extended version of this manuscript~\cite{BakK:23a}.}


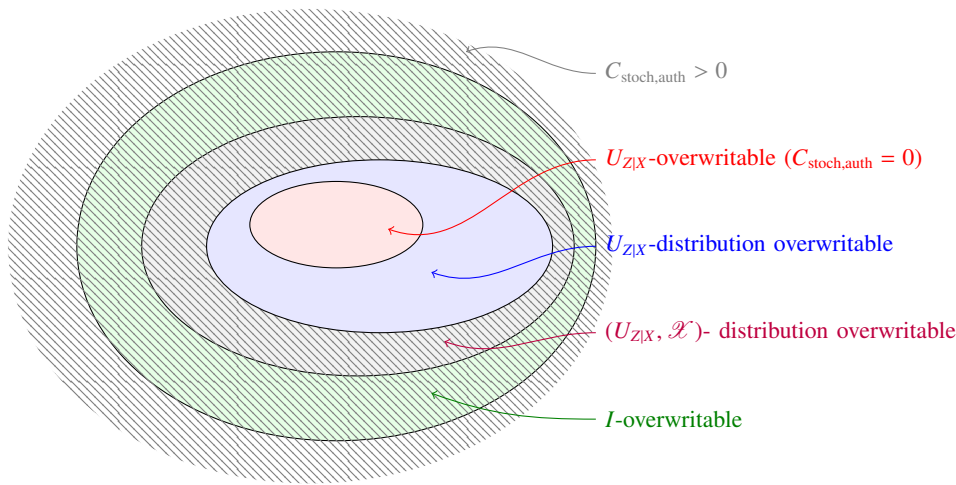
\begin{figure}
 	\centering
	\begin{tikzpicture}[scale=\figscale*0.23]
	\draw[fill=green!10!white] (0,0) ellipse (6 and 4.5);
	\draw[fill=gray!10!white] (.5,0) ellipse (5 and 3);
	\draw[color = white,thin,pattern=north west lines, pattern color=black!50!white] (-0.6,0) ellipse (7 and 5.5);
	
	\draw[fill=blue!10!white] (1,0) ellipse (4 and 2);
	
	\draw[fill=red!10!white] (0,.5) ellipse (2 and 1);
	\coordinate (achievein) at (3,4.5);
	\coordinate (achieveout) at (6,4);
	
	\coordinate (uowin) at (1.2,.4);
	\coordinate (uowout) at (6,2);
	\coordinate (udowin) at (2.2,-0.6);
	\coordinate (udowout) at (6,0);
	\coordinate (uxowin) at (1.5,.-1.2);
	\coordinate (uxowout) at (6,-2);
	\coordinate (iowin) at (2.2,-3.4);
	\coordinate (iowout) at (6,-4);
\draw[gray,<-]  (achievein) to[out=20,in=180] (achieveout)  node[right] {\small $\cst>0$};
	
\draw[red,<-]  (uowin) to[out=-20,in=180] (uowout)  node[right] {\small $\uz$-overwritable ($\cst=0$)};
\draw[blue,<-]  (udowin) to[out=-20,in=180] (udowout)  node[right] {\small $\uz$-distribution overwritable};
\draw[purple,<-]  (uxowin)+(1,-1) to[out=-20,in=180] (uxowout)  node[right] {\small $(\uz,\cX)$- distribution overwritable};
\draw[green!50!black,<-]  (iowin) to[out=-20,in=180] (iowout)  node[right] {\small $I$-overwritable};
        
\end{tikzpicture}
\caption{The hierarchy of overwritable channels is shown here. Theorem~\ref{thm:achievability} shows that positive rates of authentication are supported with stochastic codes by all channels that are in the shaded region.
}\label{fig:hierarchy}
\end{figure}

\section{Model and Notation}\label{sec:setup}
\paragraph*{Notation} We follow standard notation for information theoretic quantities (\emph{c.f.}~\cite{CK11}). We denote sets by calligraphic symbols ($\cX$,$\cY$ etc). $\mathcal{P}(\cA)$ denotes the set of all probability distributions defined on a set $\cA$. $\mathbb{V}(P,Q)$ denotes the variational distance between two probability distributions. All logarithms are to base $2$.
\paragraph{Channel} We consider \emph{myopic arbitrarily varying channels (AVCs)} that consist of pairs of channels $(\wy,\uz)$. The transmitter and the receiver are connected through the \emph{main channel} $\wy$ with input alphabet $\cX$, output alphabet $\cY$, and state set $\cS$. The state set contains a ``no-adversary'' state $\szero$, which is the default channel state when there is no adversary. The adversary, if present, is connected to the transmitter through the channel $\uz$ with input alphabet $\cX$ and output alphabet $\cZ$. As is standard in myopic AVCs, we assume that conditioned in the channel input $X$ and the channel state $S$, the main channel output $Y$ and the adversary's channel output $Z$ are independent, {\em i.e.}, the Markov chain $Y-(X,S)-Z$ holds. All channel alphabets and the state set are finite sets. 
\paragraph{Codes} We consider stochastic codes for two problems in this paper. The first problem we consider is that of communicating a message in an authenticated manner over the channel $(\wy,\uz)$ using \emph{authentication codes}.
\begin{definition}[$(N,n)$-Authentication Codes] An $(N,n)$-authentication code for a channel $(\wy,\uz)$ consists of a stochastic encoder $\authenc$ that maps messages $m\in[N]$ to codewords $X^n\in\cX^n\sim \Prob_{X^n|M=m}$ and a  deterministic decoder $\Phi:\cY^n\to[N]\cup\set{\reject}$ . 
\end{definition}

The second problem we consider is a standalone authentication problem, wherein, the goal is to authenticate a candidate message already available to the decoder by sending a tag over the channel $(\wy,\uz)$ using \emph{authentication tags}.

\begin{definition}[$(N,n)$-Authentication Tags] An $(N,n)$-authentication tag for a channel $(\wy,\uz)$ consists of a stochastic encoder $\authenc$ that maps messages $m\in[N]$ to codewords $X^n\in\cX^n\sim \Prob_{X^n|M=m}$ and a collection of deterministic decoders $\authdec_{\hat{m}}:\cY^n\to\set{\accept,\reject}$ for every $\hat{m}\in[N]$. 
\end{definition}

\paragraph{Adversarial strategies} The adversary (if present) first independently and non-causally observes $z^n$, the output of the channel $\uz$ when $x^n$ is the input. Next, the state sequence $s^n$ is selected based on the knowledge of the code and the observation $z^n$ using a strategy $Q_{S^n|Z^n}$.

\paragraph{Error Probabilities} The error probability for an authentication code is the larger of following two probabilities: \emph{(i)} the maximal  probability of  decoding to either $\reject$ or to a message other than the true message $m$ when there is no adversary, and \emph{(ii)} the probability of the decoding to neither the true message $m$ nor to $\reject$ when there is an adversary present. Note that  when there is an adversary, it is acceptable to output a $\reject$ instead of the true message $m$. 
\begin{definition}[Error Probability for Authentication Codes] We say that an $(N,n)$-authentication code $(\Prob_{X^n|M},\Phi)$ achieves error probability $\epsilon$ if 
\begin{itemize}
\item[\emph{A.}]\emph{No Adversary}\vspace{-0.5em}
\[\max_{m\in [N]}\Pr_{X^n,Y^n}(\Phi(Y^n)\neq m \vert m\mbox{ sent},S^n=\szero^n)\leq \epsilon,\ \mathrm{and}\]\vspace{-0.5em}
\item[\emph{B.}] \emph{Adversary Present}\vspace{-0.5em}
\[\max_{m\in [N]}\sup_{\Probc_{S^n|Z^n}}\Pr_{X^n,Y^n,Z^n,S^n}(\Phi(Y^n)\notin\set{m,\reject} m\vert m\mbox{ sent})\leq \epsilon.\]\vspace{-0.5em}
\end{itemize}
\end{definition}
There are two kinds of error probabilities for authentication tags: \emph{(i)} the maximal  probability of $\reject$ the true message $m$ when there is no adversary at all and \emph{(ii)} the maximal  probability of accepting a wrong candidate message when there is an adversary present. Note that when there is no adversary, we don't consider the event that a wrong candidate message is accepted in our definition of error probabilities.  
\begin{definition}[Error Probabilities for Authentication Tags] We say that an $(N,n)$-authentication tag $(\Prob_{X^n|M},\set{\authdec_m}_{m\in[N]})$ achieves  error probabilities $(\lambda_1,\lambda_2)$ if 
\begin{itemize}
\item[\emph{A.}]\emph{False Alarm:}\vspace{-0.5em}
\[\max_{m\in [N]}\Pr_{X^n,Y^n}(\authdec_m(Y^n)=\reject\vert m\mbox{ sent},S^n=\szero^n)\leq \lambda_1,\ \mathrm{and}\]\vspace{-0.5em}
\item[\emph{B.}] \emph{Missed Detection}\vspace{-0.5em}
\[\max_{\substack{m\in [N]\\ \hat{m}\in[N]\setminus\set{m}}}\sup_{\Probc_{S^n|Z^n}}\Pr_{X^n,Y^n,Z^n,S^n}(\authdec_{\hat{m}}(Y^n)=\accept\vert m\mbox{ sent})\leq \lambda_2.\]\vspace{-1em}
\end{itemize}
\end{definition}
\begin{remark} When the channel $\wy$ is non-adversarial, \emph{i.e.}, $\cS=\set{\szero}$ an authentication tag for $(\wy,\uz)$ is equivalent to be thought of as an identification code for the identification problem~\cite{AhlswedeD:89} for the channel $\wy$.
\end{remark}

\paragraph{Capacity} The authentication capacity for stochastic codes is defined as follows.
\begin{definition}[Authentication capacity] The authentication capacity $\cst$ for a channel $(\wy,\uz)$ with stochastic codes is defined to be the supremum over all $R$ such that given any  $\epsilon>0$, there is a sequence of $(N_n,n)$-authentication codes achieving error probability $\epsilon$, and $\liminf_{n\to\infty} \frac{1}{n}\log_2{N_n}\geq R$.
\end{definition}
Coming to authentication tags, it turns out that, similarly to identification codes~\cite{AhlswedeD:89}, the number of messages for authentication tags grow doubly exponentially in the blocklength.  \begin{definition}[Authentication tag capacity] The authentication tag capacity $C_{\textrm{tag}}$ for a channel $(\wy,\uz)$ is defined to be the supremum over all $R$ such that given any  $\lambda_1,\lambda_2\in (0,1)$, there is a sequence of $(N_n,n)$-authentication tags achieving error probabilities $(\lambda_1,\lambda_2)$, and $\liminf_{n\to\infty} \frac{1}{n}\log_2\log_2{N_n}\geq R$.
\end{definition}

\section{Equivalence between Authentication Code and Authentication Tag capacities}
In this section, we prove Theorem~\ref{thm:dichotomy} and show that the capacities for the authentication problem and the authentication tag problems are identical when stochastic codes are permitted. Further, whenever these capacities are positive, they equal the no-adversary capacity of the channel.
\subsection{Proof of Theorem~\ref{thm:dichotomy}} The proof of this theorem relies on Lemmas~\ref{lem:ctaggreater} and~\ref{lem:dichotomy} and by noting that $\Ctag$ is always upper bounded by the identification capacity of the channel $\mathsf{W}^{(\szero)}$.

\begin{lemma}$C_{\textrm{tag}}\geq C_{\textrm{stoch,auth}}$.\label{lem:ctaggreater}
\end{lemma}

\begin{proof}[Proof sketch] The proof proceeds along the lines of~\cite{SangwanBDP:19ISIT}. The key idea here is to first construct an identification code~\cite{AhlswedeD:89} for a noiseless channel and then transmit the codewords from the identification code using an authentication code for this channel. We furnish the details in the \app.
\end{proof}
\begin{lemma}\label{lem:dichotomy}
Whenever $C_{\mathrm{tag}}> 0$, $C_{\mathrm{stoch,auth}}=C_{\szero}$.	
\end{lemma}

\proof We first note that $C_{\mathrm{stoch,auth}}$ is always bounded form above by $C_{\szero}$ as the latter is the capacity of the channel $\wy$ when there is no adversary. 

In the following we prove that any rate smaller than $C_{\szero}$ is achievable whenever $C_{\mathrm{tag}}$ is positive.  Our achievability relies on a two-phase scheme. The first phase consists of a code for the (non-adversarial) channel $\wy(\cdot|\cdot,\szero)$. The second phase is a short phase used to verify that the message has been decoded correctly in the first phase. This architecture has been previously been shown to be capacity-achieveing  in~\cite{SangwanBDP:19ISIT} for channels with oblivious adversaries. In the following, we prove that this property continues to hold even when the adversary has a noisy view of the input codeword. 
 
 Suppose that $C_{\mathrm{tag}}>0$. Let $R<C_{\szero}$. Let $\epsilon>0$ be the target probaility of error for the authentication code. For every $n>0$, Let $t_n=\lceil (\log nR)/C_{\mathrm{tag}}\rceil$. Choose $n$ large enough such that the following are ensured:
 \begin{itemize}
 \item[\em a)] There exists a $(2^{nR},n-t_n)$  channel code with encoder $f:[2^{nR}]\to\cX^{n-t_n}$ and decoder $g:\cY^{n-t_n}\to[2^{nR}]$ for the (non-adversarial) channel $\mathsf{W}_{Y|X,S=\szero}$ with maximal error probability at most $\epsilon/2$.	
 \item[\em b)]There exists a $(2^{nR},t_n)$ authentication tag $(\hat{\Prob}_{X^{t_n}|M},\set{\phi_m})$ for $(\wy,\uz)$ achieving error probabilities $(\epsilon/2,\epsilon/2)$.
  \end{itemize}
  
  We form a $(2^{nR},n)$ authentication code $(\Prob_{X^n|M},\Phi)$ by concatenating the above codes as follows. The encoder transmits the channel code followed by the authentication tag, \emph{i.e.}, 
  \[\Prob_{X^n|M}(x^n|m)= \begin{cases}\hat{\Prob}_{X^{t_n}|M}(x_{n-t_n+1}^n|m)& \mbox{if }x^{n-t_n}=f(m),\\ 0&\mbox{otherwise.}\end{cases}\]
  The decoder first decodes the message from the channel code and then authenticates it using the authentication tag. Let $\hat{m}\triangleq g(y^{n-t_n})$. The decoder $\Phi$ is defined as
  \[\Phi(y^n)=\begin{cases}
\hat{m}& \mbox{if }\phi_{\hat{m}}(y_{n-t_n+1}^n)=\accept\\
\reject & \mbox{otherwise}.	
\end{cases}
\]
We argue that the authentication code thus constructed achieves an error probability of $\epsilon$. Consider the following two cases.
\paragraph{Case 1. No adversary} When there is no adversary, the error event is contained in the union of the event that the channel code $(f,g)$ decodes to an incorrect message and the event that the authentication tag $(\hat{\Prob}_{X^{t_n}|M},\set{\phi_m})$ rejects a correct message from the channel code. Thus, \ifthenelse{\boolean{arxiv}}{
\begin{align}
\Perr^{(\mathrm{no-adv})}&\leq \max_{m\in[2^{nR}]}\bigg[\Pr(g(Y^{n-t_n})\neq m|m\mbox{ sent}, S^{n-t_n}=\szero^{n-t_n})+ \\& \ 	\Pr(\phi_m(Y^{n-t_n+1})_m=\reject|m\mbox{ sent},S_{n-t_n+1}^n=\szero_{n-t_n+1}^n)\bigg]\\
&\leq \epsilon/2+\epsilon/2=\epsilon.
\end{align}}{$\Perr^{(\mathrm{no-adv})}<\epsilon/2+\epsilon/2=\epsilon$.}
\paragraph{Case 2. Adversary present}  In this case, the error probability can be expressed as follows.  In the following, let $l=n-t_n$ for ease of notation.
{\begin{align}
\lefteqn{\max_{m}\sup_{\Probc_{S^n|Z^n}}\Pr_{X^n,Y^n,Z^n,S^n}(\Phi(Y^n)\notin\set{M,\reject}|m\textrm{ sent})}\\
& = \max_{m} \sup_{\Probc_{S^n|Z^n}}\bigg[\Pr(\Phi(Y^n)\notin\set{m,\reject}, g(Y^{l})=m|m\textrm{ sent})\\
&\quad +\Pr(\Phi(Y^n)\notin\set{m,\reject}, g(Y^{l})\neq m|m\textrm{ sent})\bigg]\\
& \leq \max_m\sup_{\Probc_{S^n|Z^n}}\Pr(\Phi(Y^n)\notin\set{m,\reject}, g(Y^{l})=m|m\textrm{ sent})\ +\\
&\quad \max_m\sup_{\Probc_{S^n|Z^n}}\Pr(\Phi(Y^n)\notin\set{m,\reject}, g(Y^{l})\neq m|m\textrm{ sent})\label{eq:pesum}
\end{align}}
The first term in~\eqref{eq:pesum} corresponds to the channel code decoding to the correct message. In this case,  
the decoder $\Phi$ outputs either the correct message or $\reject$, both of which are acceptable outcomes when an adversary is present. Thus,
\begin{equation*}
	\Pr_{X^n,Y^n,Z^n,S^n}(\Phi(Y^n)\notin\set{m,\reject}, g(Y^{l})=m|m\textrm{ sent})=0.
\end{equation*} 
Next, we analyze the second term in~\eqref{eq:pesum}. \ifthenelse{\boolean{arxiv}}{We perform the error analysis by  allowing the adversary to choose their strategy based on the true message $m$ and the channel code output $\hat{m}$. We argue that, in this setting, even when the adversary knows $m$ and $\hat{m}$, it is sufficient for the adversary to choose a strategy based on the subset of observations $Z_{l+1}^n$ (in addition to $m$ and $\hat{m}$) rather than all of $Z^n$.The probability of not rejecting a wrong message is upper bounded as follows.}{}
\begin{align}
\lefteqn{\max_m\sup_{\Probc_{S^n|Z^n}}\Pr(\Phi(Y^n)\notin\set{m,\reject}, g(Y^{l})\neq m|m\textrm{ sent})}\\
	&\label{eq:pe}\leq  \max_{m,\hat{m}\neq m}\sup_{\Probc_{S^n|Z^n}}\Pr_{X^n,Y^n,Z^n,S^n}(\phi_{\hat{m}}(Y_{l+1}^n)=\accept|m\mbox{ sent})\\
	&= \max_{m,\hat{m}\neq m}\sup_{\Probc_{S^n|Z^n}} \sum_{\substack{x^n,y^n,z^n,s^n\\ \mathrm{s.t. } \phi_{\hat{m}}(y_{l+1}^n)=\accept}}  \Prob_{X^n|M}(x^n|m)\uz(z^n|x^n) \\ 
	& \qquad\qquad\quad\qquad\qquad\qquad\qquad \times \Probc_{S^n|Z^n}(s^n|z^n)\wy(y^n|x^n,s^n)\\
&= \max_{m,\hat{m}\neq m}\sup_{\Probc_{S^n|Z^n}} \sum_{\substack{x^n,y^n,z^n,s^n\\ \mathrm{s.t. } x^{l}=f(m)\\ \phi_{\hat{m}}(y_{l+1}^n)=\accept}}  \hat{\Prob}_{X^{t_n}|M}(x_{l+1}^n|m)\uz(z^n|x^n) \\ 
	& \qquad\quad\qquad\qquad\qquad\qquad \times \Probc_{S^n|Z^n}(s^n|z^n)\wy(y^n|x^n,s^n)\\
&= \max_{m,\hat{m}\neq m}\sup_{\Probc_{S^n|Z^n}} \sum_{\substack{x_{l+1}^n,y_{l+1}^n,z^n,s^n\\ \mathrm{s.t. } \phi_{\hat{m}}(y_{l+1}^n)=\accept}}  \hat{\Prob}_{X^{t_n}|M}(x_{l+1}^n|m) \\
& \qquad\qquad \times \uz(z^l|f(m))\uz(z_{l+1}^n|x_{l+1}^n)\Probc_{S^{n}|Z^n,S^n}(s^n|z^n)\\ 
& \label{eq:probcfull}\qquad\qquad  \times \wy(y_{l+1}^n|x_{l+1}^n,s_{l+1}^n).
\end{align}
In the above,~\eqref{eq:pe} follows from our design of the authentication code as a two-phase code (see~\applem{lem:errorprob} for details). Note that~\eqref{eq:pe} does not require a union bound over all $\hat{m}\neq m$, rather, only bounding in terms of the worst-case $\hat{m}$ suffices. Next, defining
\[\hat{\Probc}_{S^{t_n}|Z^{t_n}}(s_{l+1}^n|z_{l+1}^n)=\sum_{z^l,s^l}\uz(z^l|f(m))\Probc_{S^{n}|Z^n,S^n}(s^n|z^n),\]
the expression in \eqref{eq:probcfull} may be rewritten as 
\begin{align}
	&\max_{m,\hat{m}\neq m}\sup_{\hat{\Probc}_{S^{t_n}|Z^{t_n}}} \sum_{\substack{x_{l+1}^n,y_{l+1}^n,z_{l+1}^n,s_{l+1}^n\\ \mathrm{s.t. } \phi_{\hat{m}}(y_{l+1}^n)=\accept}}  \hat{\Prob}_{X^{t_n}|M}(x_{l+1}^n|m) \hat{\Probc}_{S^{t_n}|Z^{t_n}}(s_{l+1}^n|z_{l+1}^n)\\
	&\qquad\times  \wy(y_{l+1}^n|x_{l+1}^n,s_{l+1}^n)\\
	& = \max_{m,\hat{m}\neq m}\sup_{\hat{\Probc}_{S^{t_n}|Z^{t_n}}} \Pr_{X^{t_n},Y^{t_n},Z^{t_n},S^{t_n}} (\phi_{\hat{m}}=\accept|m\mbox{ sent})\leq \epsilon/2.\IEEEQEDhere
\end{align}
\begin{remark} The proof of Lemma~\ref{lem:dichotomy} suggests a natural two-phase architecture for authentication. The  the first phase may be thought of as the ``payload'' which can be transmitted using any reliable code for the channel $\mathsf{W}^{(\szero)}$ without adversary. The second phase is a short authentication phase. Our analysis shows that the length of the second phase need  be no larger than logarithmic in the overall block length to achieve vanishing probability of error at rates achieving the capacity. This aspect of communication with authentication has been  noted in different adversarial models in prior work~\cite{SangwanBDP:19ISIT}.
\end{remark}
\begin{remark} In the proof of Lemma~\ref{lem:dichotomy}, the no-adversary case proceeds identically to the setting with an oblivious adversary~\cite{SangwanBDP:19ISIT}. However, unlike in the oblivious setting, when adversary is present, they can correlate the attack strategy in the first and the second phases. Thus, it is not \emph{a priori} clear if the authentication tag phase can be analyzed separately from the communication phase. Our proof shows that even though the adversary may choose a strategy that depends on the entire transmission $\Probc_{S^n|Z^n}$, such a strategy is not more powerful than strategies $\hat{\Probc}_{S^{t_n}|Z^{t_n}}$ that depend only on the authentication tag phase and the knowledge of the message and the first phase reconstruction.
\end{remark}

\section{Overwritability with Stochastic Codes}

We say that $\Prob_1,\Prob_2\in\mathcal{P}(\cX)$ are \emph{myopically indistinguishable} if \[\sum_{x\in\cX}\Prob_1(x)\uz(z|x)=\sum_{x\in\cX}\Prob_2(x)\uz(z|x) \ \forall z\in\cY.\] Let $\fP_{\mathrm{ind}}$ be the partition of $\mathcal{P}(\cX)$ into equivalence classes formed by grouping indistinguishable distributions.  
\begin{definition}[\distribution overwritability]\label{def:udistow} We say that $\wy$ is \distribution overwritable if for every $\cP\in\fP_{\mathrm{ind}}$ and $x'\in\cX$, there is an adversarial strategy $\Probc_{S|Z}$ such that, for every $\Prob_X\in\cP$ and $y\in\cY$, 
\begin{equation}
\expect{X,Z,S}{\wy(y|X,S)}=\wy(y|x',\szero),	
\end{equation}
where, the expectation is with respect to the joint distribution $\Prob_{X,Z,S}= \Prob_X \uz \Probc_{S|Z}$.
\end{definition}
\begin{remark} \label{rem:nonowchar}It follows from the above definition that when $\wy$ is not \distribution overwritable, there exist $\cP\in\fP_{\mathrm{ind}}$ and $x'\in\cX$ such that for all $\Probc_{S|Z}$, there exists $\Prob_X\in\cP$ with 
\begin{equation}
\mathbb{V}\left(\expect{X,Z,S}{\wy(\cdot|X,S)},\wy(\cdot|x',\szero)\right)>0.
\end{equation}
Further, since we restrict our attention to finite alphabet channels with finite state spaces, and since every $\cP\in\fP_{\mathrm{ind}}$ is a compact subset of $\mathcal{P}(\cX)$, there exists $\nu>0$ such that 
\begin{equation}\label{eq:nonow}
\max_{\substack{\cP\in\fP_{\mathrm{ind}}\\ x'\in\cX}}\min_{\Probc_{S|Z}}\max_{\Prob_X\in\cP}\mathbb{V}\left(\expect{X,Z,S}{\wy(\cdot|X,S)},\wy(\cdot|x',\szero)\right)>\nu.
\end{equation}
\end{remark}
\begin{proof}[Proof of Theorem~\ref{thm:achievability}] We prove that $\cst>0$ for channels that are not $\uz$-distribution writable by showing the existence of authentication tags with positive rates for such channels and invoking Theorem~\ref{thm:dichotomy}. Suppose that $(\wy,\uz)$ satisfy~\eqref{eq:nonow} for some $\nu>0$ and let $(\cP,x')$ be a pair achieving the maxima in~\eqref{eq:nonow}.

\ifthenelse{\boolean{arxiv}}{\paragraph*{\underline{Case 1: $|\cP|=1$}} Let $\cP=\set{P_X}$. The result for this case follows from~\applem{lem:uzpow} by noting that, in this case, $\wy$ is $(\uz,\cX)$-overwritable (as stated in~\appdef{def:upow}). 

 \paragraph*{\underline{Case 2: $|\cP|>1$}} Now, we extend the result to the case when the set $\cP$ achieving the maximum contains more than one element.}{} First, fix $\delta>0$ and let $\cP_{\delta}$ be a $\delta$-net covering $\cP$, \emph{i.e.}, for all $P\in\cP$, there is $P'\in\cP_{\delta}$ such that $\mathbb{V}(P',P)<\delta$. Note that $\cP_\delta$ may be chosen to have a finite number of elements. In particular, we can always find $\cP_{\delta}$ such that $|\cP_\delta|<\frac{1}{\delta^{|\cX|}}$. Let $\cP_{\delta}=\set{P^{(1)},P^{(2)},\ldots,P^{(K)}}$.
 
Let $\alpha,\beta>0$ be small enough so as to satisfy
\begin{equation}
	\frac{1+\beta}{(1-\alpha)(1-\beta)}\left(\mu/4+2\frac{\alpha(1+\beta)}{1-\beta}\right)\leq \frac{\mu}{2}.\label{eq:alphabeta1}
\end{equation} Pick $\mathfrak{B}=\set{\cB_m:m\in[N]}$ as per~\applem{lem:overlappingsets}. For each $k\in[K]$, let $(P^{(k)}_{X^n|M},\set{\phi^{(k)}_m})$ be a $(P^{(k)}_X,n,\mathfrak{B},\mu/4)$-authentication tag  as given in~\appdef{def:authtag}.
 
Our achievability relies on an authentication tag consisting of several sub-blocks. Consider a $(N,nL)$-authentication tag $(P_{X^{nL}|M},\set{\phi_m})$ with encoder and decoder maps defined below. For each $l\in[L]$, the encoder uniformly picks $k$ from $[K]$ and the sub-block $X_{(l-1)n}^{ln}$ according to $P_{X^n|M}^{(k)}$. Thus,
  \[P_{X^{nL}|M}(x^{nL}|m)= \prod_{l=1}^L \sum_{k=1}^K\frac{1}{K}P_{X^n|M}^{(k)}(x_{(l-1)n+1}^{ln}|m).\]
  The decoder first decodes each sub-block and outputs an $\accept$ only if, for each block, the corresponding decoder outputs $\accept$. 
\ifthenelse{\boolean{arxiv}}{ \[\phi_m(y^{nL})=\begin{cases}\accept & \mbox{if }\phi_m^{(k)}(y_{(l-1)n+1}^{ln})=\accept\ \forall\ l\in[L]\\ \reject&\mbox{otherwise}.
	\end{cases}
\]}{}Note that, the decoder doesn't know \emph{a priori} the value of $k$ for each sub-block (since the value of $k$ is randomly and privately chosen by the transmitter). In fact, the decoder $\set{\phi_m^{(k)}}$ for any $(P^{(k)}_X,n,\mathfrak{B},\mu/4)$-authentication tag only needs the knowledge of $\mathfrak{B}$ and not of $P^{(k)}_X$. Let $\kappa^{(k)}=\kappa(P_X^{(k)},x')$ (as defined in~\appeq{eq:kappa}). Note that, over the uniform choice of $k$, $\Pr(\kappa^{(k)}>\mu/2)\geq \frac{1}{K}$ as long as $\delta$ is small enough. Thus, by~\applem{lem:typicality}, with probability at least $1/K$, $\phi_{\hat{m}}(\cdot)=\reject$ when $m\neq \hat{m}$. We first let the number of sub-blocks to be large enough and then the length of each sub-block to be large enough to conclude that the error probabilities can be made as small as desired.
\end{proof} 

\newcommand{\mif}{\mathrm{if}\ }
\newcommand{\mand}{\ \mathrm{and}\ }
\section{Examples}\label{sec:examples}
\begin{definition}[($r$-overwritable BSC$(p)$,BEC$(u)$)] The $(r$-overwritable BSC$(p)$,BEC$(u)$ is channel defined with $\cX=\cY=\set{0,1}$, $\cZ=\set{0,1,E}$, $\cS=\set{\szero,0,1}$, and the transition probabilities
\begin{align}
	\wy(y|x,s)&\triangleq\begin{cases}
		1-p & \mif y=x\mand s=\szero\\
		p & \mif y\neq x\mand s=\szero\\
		1 & \mif y=x\mand s=x\\
		1-r & \mif y=x\mand s=x\oplus 1\\
		r & \mif y\neq x\mand  s=x\oplus 1 \end{cases}\\
	\uz(z|x) & \triangleq \begin{cases} 1-u & \mif z=x \\ u & \mif z=E\end{cases}			
\end{align}
	
\end{definition}
\begin{example}[\uz-distribution overwritable but not \uz-overwritable] Consider  ($r$-overwritable BSC($p$),BEC($u$))  with $p\in(0,1/2)$, $r\in(p,1-p)$ and $u>0$. First, we note that for such channels, al input distributions are myopically distinguishable (since the probability of observing a $0$ (resp. $1$) by the adversary is proportional to the probability that the input is a $0$ (resp. $1$). Next, given any input distribution $P_X=(p_0,1-p_0)$, one can show the adversary can find a $Q_{S|Z}$ satisfying the $\uz$-distribution overwritability condition. Finally, to see that the channel is not $\uz$-overwritable, suppose that the adversary intends to overwrite with $x'=0$. When the adversary observes $E$, the adversarial strategy must work regardless of the input symbol. It turns out that there is no strategy that simultaneously works for $x=0$ and $x=1$.
	
\end{example}
\begin{example}[$(\uz,\cX)$-distribution overwritable but not $\uz$-distribution overwritable] Consider the ($r$-overwritable BSC($p$),BEC($u$)) channel with $p\in(0,1/2)$, $r\in(p,1-p)$ and $u=1$. In this case, since $\uz$  outputs an $E$ with probability $1$, all input distributions are myopically indistinguishable. Following a similar reasoning as the previous example, we conclude that there is no adversarial strategy that can work for all input distributions simultaneously. On the other hand, when the adversary knows the input distribution, they can find an adversarial strategy that  can lead to the right output distribution for $(\uz,\cX)$-distribution overwritability.
	
\end{example}
\bibliographystyle{IEEEtran}
\bibliography{references.bib}

\begin{thebibliography}{10}
\providecommand{\url}[1]{#1}
\csname url@samestyle\endcsname
\providecommand{\newblock}{\relax}
\providecommand{\bibinfo}[2]{#2}
\providecommand{\BIBentrySTDinterwordspacing}{\spaceskip=0pt\relax}
\providecommand{\BIBentryALTinterwordstretchfactor}{4}
\providecommand{\BIBentryALTinterwordspacing}{\spaceskip=\fontdimen2\font plus
\BIBentryALTinterwordstretchfactor\fontdimen3\font minus
  \fontdimen4\font\relax}
\providecommand{\BIBforeignlanguage}[2]{{%
\expandafter\ifx\csname l@#1\endcsname\relax
\typeout{** WARNING: IEEEtran.bst: No hyphenation pattern has been}%
\typeout{** loaded for the language `#1'. Using the pattern for}%
\typeout{** the default language instead.}%
\else
\language=\csname l@#1\endcsname
\fi
#2}}
\providecommand{\BIBdecl}{\relax}
\BIBdecl

\bibitem{Blachman62}
N.~Blachman, ``On the capacity of a band-limited channel perturbed by
  statistically dependent interference,'' \emph{IRE Transactions on Information
  Theory}, vol.~8, no.~1, pp. 48--55, January 1962.

\bibitem{Ahlswede78}
R.~Ahlswede, ``Elimination of correlation in random codes for arbitrarily
  varying channels,'' \emph{Zeitschrift f\"ur Wahrscheinlichkeitstheorie und
  verwandte Gebiete}, vol.~44, no.~2, pp. 159--175, 1978.

\bibitem{CsiszarN88}
I.~Csisz\'ar and P.~Narayan, ``The capacity of the arbitrarily varying channel
  revisited: positivity, constraints,'' \emph{IEEE Trans. Inform. Theory},
  vol.~34, no.~2, pp. 181--193, March 1988.

\bibitem{Sarwate:10}
A.~D. Sarwate, ``Coding against myopic adversaries,'' in \emph{2010 IEEE
  Information Theory Workshop}, 2010, pp. 1--5.

\bibitem{DeyJL:15}
B.~K. Dey, S.~Jaggi, and M.~Langberg, ``Sufficiently myopic adversaries are
  blind,'' in \emph{2015 IEEE International Symposium on Information Theory
  (ISIT)}, 2015, pp. 1164--1168.

\bibitem{ZhangVJS:18}
Y.~Zhang, S.~Vatedka, S.~Jaggi, and A.~D. Sarwate, ``Quadratically constrained
  myopic adversarial channels,'' in \emph{2018 IEEE International Symposium on
  Information Theory (ISIT)}, 2018, pp. 611--615.

\bibitem{BudkuleyDJLSW:20}
A.~J. Budkuley, B.~K. Dey, S.~Jaggi, M.~Langberg, A.~D. Sarwate, and C.~Wang,
  ``Symmetrizability for myopic avcs,'' in \emph{2020 IEEE International
  Symposium on Information Theory (ISIT)}, 2020, pp. 2103--2107.

\bibitem{SimmonsCRYPTO84}
G.~J. Simmons, ``Authentication theory/coding theory,'' in \emph{Proc. Advances
  in Cryptology-CRYPTO}, 1984, pp. 411--431.

\bibitem{MaurerIT00}
U.~M. Maurer, ``Authentication theory and hypothesis testing,'' \emph{IEEE
  Trans. Inform. Theory}, vol.~46, no.~4, pp. 1350--1356, July 2000.

\bibitem{LaiEPIT09}
L.~Lai, H.~E. Gamal, and H.~V. Poor, ``Authentication over noisy channels,''
  \emph{IEEE Trans. Inform. Theory}, vol.~55, no.~2, pp. 906--916, February
  2009.

\bibitem{TuLIT18}
W.~Tu and L.~Lai, ``Keyless authentication and authenticated capacity,''
  \emph{IEEE Trans. Inform. Theory}, vol.~64, no.~5, pp. 3696--3714, May 2018.

\bibitem{KosutKITW18}
O.~Kosut and J.~Kliewer, ``Authentication capacity of adversarial channels,''
  in \emph{Proc. Information Theory Workshop}.\hskip 1em plus 0.5em minus
  0.4em\relax Guangzhou, 2018.

\bibitem{GravesYS:16}
E.~Graves, P.~Yu, and P.~Spasojevic, ``Keyless authentication in the presence
  of a simultaneously transmitting adversary,'' in \emph{2016 IEEE Information
  Theory Workshop (ITW)}, 2016, pp. 201--205.

\bibitem{GravesBKKY:23}
E.~Graves, A.~Beemer, J.~Kliewer, O.~Kosut, and P.~L. Yu, ``Keyless
  authentication for awgn channels,'' \emph{IEEE Transactions on Information
  Theory}, vol.~69, no.~1, pp. 496--519, 2023.

\bibitem{SangwanBDP:19ISIT}
N.~Sangwan, M.~Bakshi, B.~K. Dey, and V.~M. Prabhakaran, ``Multiple access
  channels with adversarial users,'' in \emph{2019 IEEE International Symposium
  on Information Theory (ISIT)}, 2019, pp. 435--439.

\bibitem{BeemerGKKY:20a}
A.~Beemer, E.~Graves, J.~Kliewer, O.~Kosut, and P.~Yu, ``Authentication and
  partial message correction over adversarial multiple-access channels,'' in
  \emph{2020 IEEE Conference on Communications and Network Security (CNS)},
  2020, pp. 1--6.

\bibitem{BeemerKKGY:19}
A.~Beemer, O.~Kosut, J.~Kliewer, E.~Graves, and P.~Yu, ``Authentication against
  a myopic adversary,'' in \emph{2019 IEEE Conference on Communications and
  Network Security (CNS)}, 2019, pp. 1--5.

\bibitem{BeemerGKKY:20}
A.~Beemer, E.~Graves, J.~Kliewer, O.~Kosut, and P.~Yu, ``Authentication with
  mildly myopic adversaries,'' in \emph{2020 IEEE International Symposium on
  Information Theory (ISIT)}, 2020, pp. 984--989.

\bibitem{AhlswedeD:89}
R.~Ahlswede and G.~Dueck, ``Identification via channels,'' \emph{IEEE Trans.
  Inform. Theory}, vol.~35, no.~1, pp. 15--29, January 1989.

\bibitem{CK11}
I.~Csisz\'ar and J.~K\"orner, \emph{Information Theory: Coding Theorems for
  Discrete Memoryless Systems}.\hskip 1em plus 0.5em minus 0.4em\relax
  Cambridge University Press, 2011.

\end{thebibliography}
\ifthenelse{\boolean{arxiv}}{}{\clearpage}
\appendices
\section{Proof of Lemma~\ref{lem:ctaggreater}}
\begin{proof}	Let $\mathsf{W}_{id}$ denote the noiseless binary channel whose input equals the output with probability one. Fix $R<C_{\textrm{stoch,auth}}$. Let $\hat{R}\in(R,C_{\textrm{stoch,auth}})$ and $\tilde{R}=R/\hat{R}$. Let $\lambda_1,\lambda_2\in(0,1/2)$ be desired misauthentication probabilities. 
Following~\cite[Theorem~1]{AhlswedeD:89}, for all $\tilde{n}$ large enough, there exist $(N^{\mathrm{(id)}},\tilde{n})$-identification codes $\left(\tilde{\Prob}_{B^{\tilde{n}}|M},\set{	\tilde{\phi}_m}_{m\in[N^{\mathrm{(id)}}]}\right)$ for  channel $\mathsf{W}_{id}$ that achieves misidentification probabilities $(\lambda_1/2,\lambda_2/2)$  and $N^{\mathrm{(id)}}\geq \lfloor 2^{2^{\tilde{n}\tilde{R}}}\rfloor$.  Note that the number of codewords for such a code is at most $N^{\mathrm{(id)}}_{out}\triangleq 2^{\tilde{n}}$. Next, from the definition of authentication capacity, for all $n$ large enough, there exist $(N^{\mathrm{(auth)}},n)$-authentication codes $\left(\hat{\Prob}_{X^{n}|\mathbf{B}},\hat{\Phi}\right)$ for the channel $(\wy,\uz)$ with error probability at most $\max\set{\lambda_1/2,\lambda_2/2}$ and $N^{\mathrm{(auth)}}\geq \lfloor 2^{n \hat{R}}\rfloor$. 

Now, we compose the two codes in the following manner. First, let $n$ be large enough so that there exist both an $(N^{\mathrm{(id)}},\tilde{n})$-identification code and an $(N^{\mathrm{(auth)}},n)$-authentication code of the form above with $N^{\mathrm{(auth)}}=N^{\mathrm{(id)}}_{out}$.  Consider an $(N^{\mathrm{(id)}},n)$-authentication tag for the channel $(\wy,\uz)$ defined through the encoder 
\[\Prob_{X^n|M}(x^n|m)=\sum_{b^{\tilde{n}}\in\set{0,1}^{\tilde{n}}}\hat{\Prob}_{X^n|\mathbf{B}}(x^n|b^{\tilde{n}})\tilde{\Prob}_{B^{\tilde{n}}|M}(b^{\tilde{n}}|m),\]
	and the decoders
	\[\phi_{\hat{m}}(y^n)=\begin{cases}\reject& \mbox{if } \Phi(y^n)=\reject \\
	\reject & \mbox{if } \Phi(y^n)=\hat{\mathbf{b}}\in\{0,1\}^{\tilde{n}}\mbox{ and }\\ &\quad \tilde{\phi}_{\hat{m}}(\mathbf{b})=\reject\\
	\accept&\mbox{otherwise},\end{cases}\]
	for all $m,\hat{m}\in[N^{\mathrm{(id)}}]$, $x^n\in\cX^n$, and $y^n\in\cY^n$.
	We note that the rate of this code is 
	\begin{align*}
 	\frac{1}{n}\log\log N^{\mathrm{(id)}} &\geq \frac{\tilde{n}{\tilde{R}}}{n}\\
 	&\geq \frac{\tilde{n}{\tilde{R}}}{(1/\hat{R})\log\tilde{N}} = \tilde{R} \hat{R}\\
 	&=R.
 \end{align*}
Further, the error probabilities for the authentication tag $(\tilde{\Prob}_{X^n|M},\set{\phi_{\hat{m}}})$ are bounded from above by the sum of the corresponding misidentification probabilities for the identification code $(\Prob_{B^{\tilde{n}}|M},\set{\tilde{\phi}_m})$ and the error probability for the authentication code $(\hat{\Prob}_{X^n|\mathbf{B}},\hat{\Phi})$. This shows that $C_{\mathrm{tag}}\geq C_{\mathrm{stoch,auth}}$.
\end{proof}
\section{}
\begin{lemma}\label{lem:errorprob} Let $\Phi$, $g$ and $\set{\phi_m}_{m}$ be defined as in the proof of Lemma~\ref{lem:dichotomy}. Then, we have, 
\begin{align*}
\lefteqn{\max_m\sup_{\Probc_{S^n|Z^n}}\Pr(\Phi(Y^n)\notin\set{m,\reject}, g(Y^{l})\neq m|m\textrm{ sent})}\\
	&\leq  \max_{m,\hat{m}\neq m}\sup_{\Probc_{S^n|Z^n}}\Pr_{X^n,Y^n,Z^n,S^n}(\phi_{\hat{m}}(Y_{l+1}^n)=\accept|m\mbox{ sent})
\end{align*}
	
\end{lemma}
\begin{proof} Consider the following chain of inequalities.
	\begin{align}
	\lefteqn{\Pr(\Phi(Y^n)\notin\set{m,\reject}, g(Y^{l})\neq m|m\textrm{ sent})\label{eq:step1}}\\
	&= \Pr_{X^n,Y^n,Z^n,S^n}(\phi_{g(Y^l)}(Y_{l+1}^n)=\accept, g(Y^l)\neq m|m\mbox{ sent}) \label{eq:step2}\\
	&= \sum_{m'\neq m}\Pr_{X^n,Y^n,Z^n,S^n}(\phi_{m'}(Y_{l+1}^n)=\accept|g(Y^l)=m',m\mbox{ sent})\eqbr{\qquad\times} \Pr_{X^n,Y^n,Z^n,S^n}(g(Y^l)=m'|m\mbox{ sent}) \label{eq:step3}\\
	&= \sum_{m'\neq m}\Pr_{X^n,Y^n,Z^n,S^n}(\phi_{m'}(Y_{l+1}^n)=\accept|m\mbox{ sent})\eqbr{\qquad\times} \Pr_{X^n,Y^n,Z^n,S^n}(g(Y^l)=m'|m\mbox{ sent}) \label{eq:step4}\\
	& \leq \max_{\hat{m}\neq m}\Pr_{X^n,Y^n,Z^n,S^n}(\phi_{\hat{m}}(Y_{l+1}^n)=\accept|m\mbox{ sent})\eqbr{\qquad\times} \sum_{m'}\Pr_{X^n,Y^n,Z^n,S^n}(g(Y^l)=m'|m\mbox{ sent}) \label{eq:step5}\\
	& \leq \max_{\hat{m}\neq m}\Pr_{X^n,Y^n,Z^n,S^n}(\phi_{\hat{m}}(Y_{l+1}^n)=\accept|m\mbox{ sent}).
\end{align}

In the above,~\eqref{eq:step2} follows from the definition of $\Phi$, and~\eqref{eq:step4} is obtained by noting that $Y^l-(X^n,Z^n,S^n,M)-Y_{l+1}^n$ is a Markov chain due to our two phase coding scheme. Finally, taking the suprema with respect to $m$ and $Q_{S^n|Z^n}$ gives us the bound in the lemma statement. 
\end{proof}
\section{Overwritability conditions}
\begin{definition}[$\uz$-overwritability~\cite{BeemerKKGY:19}] \label{def:uow} We say that $\wy$ is \uz-overwritable if, for every $x'\in\cX$, there exists an adversarial strategy $\Probc_{S|Z}$ such that for every $x\in\cX$ and $y\in\cY$,
\[\expect{Z,S}{\wy(y|x,S)}=\wy(y|x',\szero),\]
	where, the random variables $(Z,S)$ are distributed according to the joint distribution $\uz(\cdot|x)\Probc_{Z|S}(\cdot|\cdot)$.
\end{definition}
%
\begin{definition}[$I$-overwritability~\cite{BeemerGKKY:20}] \label{def:iow} We say that $\wy$ is $I$-overwritable if, for every $x,x'\in\cX$, there exists an adversarial strategy $\Probc_{S|Z}$ such that for every $y\in\cY$,
\[\expect{Z,S}{\wy(y|x,S)}=\wy(y|x',\szero),\]
	where, the random variables $(Z,S)$ are distributed according to the joint distribution $\uz(\cdot|x)\Probc_{Z|S}(\cdot|\cdot)$.
\end{definition}

\begin{definition}[$(\uz,\cX)$-distribution overwritability]\label{def:upow} We say that $\wy$ is $(\uz,\cX)$-distribution overwritable if for every $P\in\mathcal{P}(\cX)$ and $x'\in\cX$, there is an adversarial strategy $\Probc_{S|Z}$ such that, for every $y\in\cY$, 
\begin{equation}
\expect{X,Z,S}{\wy(y|X,S)}=\wy(y|x',\szero),	
\end{equation}
where, the expectation is with respect to the joint distribution $\Prob_{X,Z,S}= \Prob_X \uz \Probc_{S|Z}$.
\end{definition}
\begin{definition}[$(P_X,n, \mathfrak{B},\mu)$-authentication tag] \label{def:authtag} Given $P_X\in\mathcal{P}(\cX)$, blocklength $n$,  a family of subsets of $\mathfrak{B}$ of $[n]$,  and a decoding threshold $\rho$, a $(P_X,n, \mathfrak{B},\mu)$-authentication tag is defined through the encoder map
\[\Prob_{X^n|M}(x^n|m)=\prod_{i\notin\cB}\Prob_X(x)\prod_{i\in \cB}\mathbbm{1}_{x'}(x), \quad x^n\in\cX^n,\]
 and the decoder maps $\set{\phi_m}$ specified below.
Given $y\in\cY$, $y^n\in\cY^n$, and $\cB\in\mathfrak{B}$, let \[\hat{P}_{y^n}(y|\cB)\triangleq \frac{\{i\in\cB:y_i=y\}}{|\cB|}.\]
For each $m\in[N]$, and  $y^n\in\cY^n$ the decoder $\phi_m$ outputs according to the following rule
\[\phi_m(y^n)=\begin{cases}
	\accept & \mathrm{if}\ \mathbb{V}(\hat{P}_{y^n}(\cdot|\cB_{m}),\wy(\cdot|x',\szero))<\rho\\
	\reject&\mathrm{otherwise.} 
\end{cases}\]
\end{definition}
\begin{lemma}\label{lem:uzpow} If $\wy$ is not $(\uz,\cX)$-overwritable, $\cst>0$.
\end{lemma}
\begin{proof}
Following along a similar reasoning in Remark~\ref{rem:nonowchar}, we note that if $\wy$ is not  $(\uz,\cX)$-overwritable, there exist $P_X\in\mathcal{P}(\cX)$ and $x'\in\cX$ and $\mu>0$ such that 
\begin{equation}
\min_{\Probc_{S|Z}}\mathbb{V}\left(\expect{X,Z,S}{\wy(\cdot|X,S)},\wy(\cdot|x',\szero)\right)> \mu.\label{eq:nonupow}\end{equation}
We invoke the Theorem~\ref{thm:dichotomy} to note that it is sufficient to show the existence of authentication tags of positive rate for large enough blocklengths. To this end, let $(\lambda_1,\lambda_2)$ be desired false alarm and missed detection probabilities for the authentication tag. Let $\alpha,\beta>0$ be small enough so as to satisfy
\begin{equation}
	\frac{1+\beta}{(1-\alpha)(1-\beta)}\left(\mu/4+2\frac{\alpha(1+\beta)}{1-\beta}\right)\leq \frac{\mu}{2}.\label{eq:alphabeta}
\end{equation} Pick $\mathfrak{B}=\set{\cB_m:m\in[N]}$ as per Lemma~\ref{lem:overlappingsets}. Consider a $(P_X,n,\mathfrak{B},\mu/4)$-tag $(P_{X^n|M},\set{\phi_m})$ as given in Definition~\ref{def:authtag}.
We first analyze the false alarm probability.  Note that 
\begin{align}
&\Pr_{X^n,Y^n}(\authdec_m(Y^n)=\reject\vert m\mbox{ sent},S^n=\szero^n)	\\
& = \Pr_{X^n,Y^n}\left(\mathbb{V}\left(\hat{P}_{Y^n}(\cdot|\cB_m),\wy(\cdot|x',\szero\right)>\mu/4\right)\\
&= \Pr_{Y_{\cB_m}}\left(\mathbb{V}\left(\hat{P}_{Y^n}(\cdot|\cB_m),\wy(\cdot|x',\szero\right)>\mu/4 \big| x_{\cB_m}=(x')^{|\cB_m|}\right).
\end{align}
Now, noting that for each $i\in\cB_m$, $x_i=x'$, and hence, $Y_i\sim \wy(\cdot|x',\szero)$, the probability of false rejection is simply the probability that a sequence of length $|\cB_m|$ is not typical when each symbol in the sequence is drawn i.i.d. from $\wy(\cdot|x',\szero)$. Thus, by Weak Law of Large Numbers, for $n$ large enough, the false alarm probability for our construction is smaller than $\lambda_1$.

Now, we analyze the missed detection probability. To this end, note that, for any $m$, $\hat{m}$ s.t. $\hat{m}\neq m$ and any adversarial strategy $Q_{S^n|Z^n}$, we have,  
\begin{align} 
&	\Pr_{X^n,Y^n,Z^n,S^n}(\authdec_{\hat{m}}(Y^n)=\accept\vert m\mbox{ sent})\\
& = \Pr_{X^n,Y^n,Z^n,S^n}\left(\mathbb{V}\left(\hat{P}_{Y^n}(\cdot|\cB_{\hat{m}}),\wy(\cdot|x',\szero\right)<\mu/4\right)\\
& \stackrel{(a)}{\leq} \Pr_{X^n,Y^n,Z^n,S^n}\Bigg[\mathbb{V}\left(\hat{P}_{Y^n}(\cdot|\cB_{\hat{m}}\setminus\cB_{m}),\wy(\cdot|x',\szero\right)\\
&\qquad\qquad\qquad\qquad<\frac{1+\beta}{(1-\alpha)(1-\beta)}\left(\mu/4+2\frac{\alpha(1+\beta)}{1-\beta}\right)\Bigg]\\
& \stackrel{(b)}{\leq}  \Pr_{X^n,Y^n,Z^n,S^n}\Bigg[\mathbb{V}\left(\hat{P}_{Y^n}(\cdot|\cB_{\hat{m}}\setminus\cB_{m}),\wy(\cdot|x',\szero\right)<\frac{\mu}{2}\Bigg]\label{eq:lastprob}.
\end{align}
In the above, over all candidate messages $\hat{m}$ from the first phase)$(a)$ follows from Property 5) of Lemma~\ref{lem:overlappingsets}. $(b)$ follows from the condition in~\eqref{eq:alphabeta}. Finally, we note that for all $i\in\cB_{\hat{m}}\setminus\cB_m$, $X_i$ is drawn i.i.d. from $P_X$. Thus, application of Lemma~\ref{lem:typicality} gives that the probability in~\eqref{eq:lastprob} is bounded from above by $2^{-n\alpha(1-\alpha)(1+\beta)\gamma}$ for some $\gamma>0$. This proves that the missed detection probability is smaller than $\lambda_2$ for $n$ large enough.\end{proof}
\newcommand{\tc}{\mathsf{T}}
\newcommand{\Probv}{V}

For any $P_X\in\mathcal{P}(\cX)$ and $x'\in\cX$, let
\begin{align}
& \kappa(P_X,x')\triangleq \\& \min_{Q_{S|Z}}\mathbb{V}\left(\sum_{x,z,s}P_X(x)\uz(z|x)Q_{S|Z}(s|z)\wy(\cdot|x,s),\wy(\cdot|x',\szero)\right). \label{eq:kappa}
\end{align}
\begin{lemma}\label{lem:typicality} Let $(X^n,Y^n,Z^n,S^n)$ be drawn from  $\prod_i P_{X}(\cdot)\prod_{i}\uz (\cdot|\cdot) Q_{S^n|Z^n}\prod_{i}\wy(\cdot|\cdot,\cdot)$.  Then, there exists $\gamma>0$, such that for $n$ large enough and for all $\Probc_{S^n|Z^n}$, 
\[\Pr(\mathbb{V}(\hat{P}_{Y^n},\wy(\cdot|x',\szero))<\kappa(P_X,x')/2)<2^{-n\gamma}.\]
\end{lemma}

\begin{proof} 
\emph{Notation:} In the following, we borrow notation from~\cite[Chapter~1]{CK11}. We let $\cP^{(n)}=\cP^{(n)}(\cY\times\cZ\times\cS)$ be the set of all types on $\cY\times\cZ\times\cS$. For a joint type $\Prob\in\cP^{(n)}$, let $\tc_\Prob$ denotes the corresponding type class, \emph{i.e.}, all sequences whose empirical distribution is $\Prob$. Similarly, we define the set of types and type classes for all other combinations random variables.  


Suppose that $(z^n,s^n)\in \tc_{Q_{ZS}}$ for some $\Probc_{ZS}\in\cP^{(n)}(\cZ\times \cS)$. Let $V_{Y|ZS}$ be the conditional probability of $Y$ given $(Z,S)$ when $X$ is drawn as per $P_X$ and $S$ is drawn as per $Q_{S|Z}$, \emph{i.e.},\[V_{Y|ZS}(y|zs)=\frac{\sum_{x}P_X(x)\uz(z|x)Q_{S|Z}(s|z)\wy(y|x,s)}{\sum_{x}P_X(x)\uz(z|x)Q_{S|Z}(s|z)}.\]   Let $Q_{YZS}=V_{Y|ZS}Q_{ZS}$ and let $Q_Y(=Q_Y^{(z^n,s^n)})$ be the marginal distribution induced by $Q_{YZS}$ on $\cY$. Note that this distribution is completely determined by the joint type of $(z^n,s^n)$, the given channel conditional probabilities, and the given input distribution. 

\emph{Proof overview:} The proof of this lemma proceeds by first showing in~\eqref{eq:jttype} that, for any $(z^n,s^n)$,  the joint type of $(Y^n,z^n,s^n)$ is close to $Q_{YZS}=V_{Y|ZS}Q_{SZ}$,  with high probability over the generation of $Y^n$. Next, we show that, with a high probability over $Z^n$, the joint type of $(Y^n,Z^n,S^n)$ is close to the single letter distribution $V_{Y|ZS}P_ZQ_{S|Z}$, where $P_Z$ is the probability distribution of the random variable $Z$ when $X$ is drawn from $P_X$ and $Z$ is the output of the channel $\uz$. This allows us to bound the probability that the variational distance in the lemma statement exceeds $\kappa(P_X,x')/2$ by a similar expression in terms of the type $Q_Y$ of $Y$ (Eq.~\eqref{eq:lastbigeqn}). Finally, we note that $Q_Y$ approximately satisfies the form required in the definition~\eqref{eq:kappa} to conclude that the lemma statement holds.  

\emph{Proof details:} Let $\kappa=\kappa(P_X,x')$. Let $(z^n,s^n)\in \tc_{Q_{ZS}}$ for some $\Probc_{ZS}\in\cP^{(n)}(\cZ\times \cS)$ and let $Y^n\sim\prod_i\Probv_{Y|Z,S}(\cdot|z_i,s_i)$. Following~\cite[Lemma~2.6]{CK11}, we first note that, for all $\Probc'_{YZS}\in\cP^{(n)}(\cY\times\cZ\times \cS)$ such that $ \Probc'_{ZS}=\Probc_{ZS}$,\[\Pr((Y^n,z^n,s^n)\in \tc_{Q'_{YZS}})\leq 2^{-nD(\Prob'_{Y|ZS}||\Probv_{Y|ZS})|\Prob_{ZS})}.\]
	Thus, for the given $(z^n,s^n)$, and $\eta>0$,
	\begin{align}
	& \Pr( D({\hat{P}}_{Y^n,z^n,s^n}||\Probv_{Y|ZS}\Probc_{ZS})>\eta    | z^n,s^n) \\
	\quad	& = \sum_{\substack{\Probc'_{YZS}\in\cP^{(n)}(\cY\times\cZ\times\cS)\\
	s.t.\ \Probc'_{ZS}=\Probc_{ZS}\\ D(\Probc'_{YZS}||V_{Y|ZS}Q_{ZS})>\eta}}\Pr((Y^n,z^n,s^n)\in \tc_{P'_{YZS}})\\
	\quad & \leq \sum_{\substack{\Probc'_{YZS}\in\cP^{(n)}(\cY\times\cZ\times\cS)\\
	s.t.\ \Probc'_{ZS}=\Probc_{ZS}\\ D(\Probc'_{YZS}||V_{Y|ZS}Q_{ZS})>\eta}} e^{-nD(\Probc'_{Y|ZS}||\Probv_{Y|ZS}|\Probc_{ZS})}\\
	&  \leq (n+1)^{|\cY|} 2^{-n\eta} < 2^{-n\eta/2}.\label{eq:jttype}
	\end{align}
Now, by Pinsker's inequality, whenever $D(\hat{P}_{Y^n,z^n,s^n}||\Probc_{YZS})<\eta$, we have $\mathbb{V}(\hat{P}_{Y^n,z^n,s^n},\Probc_{YZS}))<\sqrt{2\eta}$. Note that 

\begin{align}
\mathbb{V}(\hat{\Prob}_{Y^n},\wy(\cdot|x',\szero))&\geq 	\mathbb{V}(Q_Y,\wy(\cdot|x',\szero))-\mathbb{V}(\hat{\Prob}_{Y^n},Q_Y)\\
& \geq  	\mathbb{V}(Q_Y,\wy(\cdot|x',\szero))\\
&\qquad\qquad\qquad-\mathbb{V}(\hat{P}_{Y^n,z^n,s^n},Q_{YZS})\\
& \geq  	\mathbb{V}(Q_Y,\wy(\cdot|x',\szero))-\sqrt{2\eta}.\label{eq:pinsker1}\end{align}
Let $P_Z$ be the distribution of each $Z_i$ under the conditions of the lemma, \emph{i.e.} $P_Z(\cdot)=\sum_{x\in\cX}P_X(x)\uz(\cdot|x)$. We have,
\begin{align}
&\Pr\bigg(\mathbb{V}(\hat{P}_{Y^n},\wy(\cdot|x',\szero))<\frac{\kappa}{2}\bigg)	\\
&= \sum_{z^n,s^n}\bigg[\Pr\bigg(\mathbb{V}(\hat{P}_{Y^n},\wy(\cdot|x',\szero))<\frac{\kappa}{2}\bigg|z^n,s^n\bigg)\\ 
&\qquad\times \prod_i\Prob_Z(z_i) Q_{S^n|Z^n}(s^n|z^n)\bigg] \\
& = \sum_{Q_{ZS}\in\cP^{(n)}} \sum_{(z^n,s^n)\in \tc_{Q_{ZS}}}\bigg[\Pr\bigg(\mathbb{V}(\hat{P}_{Y^n},\wy(\cdot|x',\szero))<\frac{\kappa}{2}\bigg|z^n,s^n\bigg) \\
&\qquad\times \prod_i\Prob_Z(z_i) Q_{S^n|Z^n}(s^n|z^n)\bigg] \\
& \stackrel{(a)}{\leq} \sum_{Q_{ZS}\in\cP^{(n)}} \sum_{\substack{(z^n,s^n)\\ \;\;\in \tc_{Q_{ZS}}}}\bigg[\Pr\bigg(\mathbb{V}(Q_Y,\wy(\cdot|x',\szero))-\sqrt{2\eta}<\frac{\kappa}{2}\bigg|z^n,s^n\bigg) \\
&\qquad\times \prod_i\Prob_Z(z_i) Q_{S^n|Z^n}(s^n|z^n)\bigg]+2^{-n\eta/2} \\
& \stackrel{(b)}{\leq} \sum_{Q_{ZS}\in\cP^{(n)}} \sum_{\substack{(z^n,s^n)\\ \;\;\in \tc_{Q_{ZS}}}}\bigg[\Pr\bigg(\mathbb{V}(Q_Y,\wy(\cdot|x',\szero))-\sqrt{2\eta}<\frac{\kappa}{2}\bigg|z^n,s^n\bigg) \\
&\qquad\times 2^{-nD(Q_Z||P_Z)-nH(Q_Z)} Q_{S^n|Z^n}(s^n|z^n)\bigg]+2^{-n\eta/2} \\
& \stackrel{(c)}{\leq} \sum_{\substack{Q_{ZS}\in\cP^{(n)}\\D(Q_Z||P_Z)<\xi}} \sum_{\substack{(z^n,s^n)\\ \;\;\in \tc_{Q_{ZS}}}}\bigg[\Pr\bigg(\mathbb{V}(Q_Y,\wy(\cdot|x',\szero))-\sqrt{2\eta}<\frac{\kappa}{2}\bigg|z^n,s^n\bigg) \\
&\qquad\times 2^{-nD(Q_Z||P_Z)-nH(Q_Z)} Q_{S^n|Z^n}(s^n|z^n)\bigg]+2^{-n\xi/2}+2^{-n\eta/2}. \label{eq:lastbigeqn}
\end{align}

In the above, $(a)$ follows from~\eqref{eq:pinsker1}. $(b)$ follows from the method of types. To obtain $(c)$, we bound the probability contribution from all $z^n$ in type classes with $D(Q_Z||P_Z)>\xi$ by $2^{-n\xi/2}$. Let $\tilde{Q}_Y$ be the marginal distribution on $\cY$ induced by the joint distribution $P_ZQ_{S|Z}V_{Y|ZS}$. Note that $\tilde{Q}_Y$ is of the form
\begin{align}
	\tilde{Q}_Y(y)&=\sum_{z,s}P_Z(z)Q_{S|Z}(s|z) V_{Y|ZS}(y|z,s)\\
	& = \sum_{x,z,s}P_X(x)\uz(z|x)Q_{S|Z}(s|z) W_{Y|XS}(y|x,s)\label{eq:tildeq}.
\end{align}

Now, note that  for all $Q_{S|Z}$, $D(Q_Z||P_Z)<\xi$ implies that $\mathbb{V}(Q_Y,\tilde{Q}_Y)<\sqrt{2\xi}$, since 
\[D(Q_Y||\tilde{Q}_Y)\leq D(Q_ZQ_{S|Z}V_{Y|ZS}||P_ZQ_{S|Z}V_{Y|ZS}) = D(Q_Z||P_Z), \] 
and Pinsker's inequality gives us $\mathbb{V}(Q_Y,\tilde{Q}_Y)\leq \sqrt{2D(Q_Y||\tilde{Q}_Y)}$. As a consequence, we have, for all $Q_Y$ that obtained from a joint distribution $Q_ZQ_{S|Z}V_{Y|ZS}$ with $D(Q_Z||P_Z)<\xi$, we have, \[\mathbb{V}(Q_Y,\wy(\cdot|x',\szero))\geq \mathbb{V}(\tilde{Q}_Y,\wy(\cdot|x',\szero))-\sqrt{2\xi}.\]
Applying this bound to~\eqref{eq:lastbigeqn} gives
\begin{align}
&\Pr\bigg(\mathbb{V}(\hat{P}_{Y^n},\wy(\cdot|x',\szero))<\frac{\kappa}{2}\bigg)\\
& \leq  \Pr\bigg(\mathbb{V}(\tilde{Q}_Y,\wy(\cdot|x',\szero))<\sqrt{2\xi}+\sqrt{2\eta}+\frac{\kappa}{2}\bigg) +2^{-n\xi/2}+2^{-n\eta/2}.
\end{align}
Finally, comparing~\eqref{eq:tildeq} with~\eqref{eq:kappa}, we note that $\tilde{Q}_Y$  satisfies $\mathbb{V}(\tilde{Q}_Y,\wy(\cdot|x',\szero))\geq \kappa$. Thus, as long as $\sqrt{2\xi}+\sqrt{2\mu}<\mu/2$, the conclusion of the lemma follows with $\gamma=\min\set{\xi,\eta}/2$.

\end{proof}

The following corollary is useful in the proof of Theorem~\ref{thm:achievability}.
\begin{corollary}
	\label{cor:typicality} Suppose that $(\wy,\uz)$ satisfy~\eqref{eq:nonupow}. Let $(X^n,Y^n,Z^n,S^n)\sim \prod_i P_{X}(\cdot)\prod_{i}\uz (\cdot|\cdot) Q_{S^n|Z^n}\prod_{i}\wy(\cdot|\cdot,\cdot)$.  Then, there exists $\gamma>0$, such that for $n$ large enough and for all $\Probc_{S^n|Z^n}$, 
\[\Pr(\mathbb{V}(\hat{P}_{Y^n},\wy(\cdot|x',\szero))<\mu/2)<2^{-n\gamma}.\]
\end{corollary}
\begin{proof}
The above follows directly from Lemma~\ref{lem:typicality} by replacing $\kappa$ by $\mu$.	
\end{proof}

\begin{lemma} \label{lem:overlappingsets}Let $n\in\mathbb{N}$ and $\alpha,\beta\in(0,1)$. Let $R<\min\set{\beta^2\alpha^2/6,\beta^2\alpha(1-\alpha)/4}$. Then, there exists a family  $\mathfrak{B}$ of subsets of $[t]$ satisfying:
\begin{enumerate}
	\item $ \alpha (1-\beta) n\leq|\cB|\leq \alpha (1+\beta) n$ for all $\cB\in\mathfrak{B}$,
	\item $|\cB\cap\cB'|<\alpha^2(1+\beta) n$ for all $\cB,\cB'\in\mathfrak{B}$,
	\item $|\cB'\setminus\cB|>\alpha (1-\alpha)(1+\beta) n$ for all $\cB,\cB'\in\mathfrak{B}$, and
	\item $|\mathfrak{B}|\geq 2^{R n}$.
	\item For every $y^n\in\cY^n$, \begin{align}
	&\mathbb{V}\left(\hat{P}_{y^n}(\cdot|\cB'),\wy(\cdot|x',\szero)\right)\\
	& \geq  \frac{(1-\alpha)(1-\beta)}{1+\beta}\mathbb{V}\left(\hat{P}_{y^n}(\cdot|\cB'\setminus\cB),\wy(\cdot|x',\szero)\right)\\
	&\qquad\qquad\qquad-2 \frac{\alpha(1+\beta)}{(1-\beta)}.
\end{align}
\end{enumerate}
\end{lemma}
\begin{proof} Choose $\mathfrak{B}$ by picking $N$ subsets of $[t]$ such that each  $\cB\in\mathfrak{B}$ by independently including every element of $[t]$ with probability $\alpha$ each. Then, we have, \begin{align}
&\expect{}{|\cB|}=\alpha n,\\
&\expect{}{|\cB\cap\cB'|}=\alpha^2n,\ \mathrm{and}\\
&\expect{}{|\cB\setminus\cB'|}=\alpha(1-\alpha)n
\end{align} 
	for all $\cB,\cB'\in\mathfrak{B}$. Next, we apply Chernoff bound to conclude that 
	\begin{align}
	& \Pr(\left| |\cB|-\alpha n\right|>\alpha\beta n)\leq 2e^{-\beta^2\alpha n/3},\\
	& \Pr(|\cB\cap\cB'|>\alpha^2(1+\beta)n)\leq e^{-\beta^2\alpha^2 n/3}, \ \mathrm{and}\\
	& \Pr(|\cB\setminus\cB'|<\alpha(1-\alpha)(1-\beta)n)\leq e^{-\beta^2\alpha(1-\alpha) n/2}	
	\end{align}
for all $\cB,\cB'\in\mathfrak{B}$. Taking a union bound over all $N$ sets in the first inequality and $N^2$ pairs of sets in the second and third, we see that this random choice of $\mathfrak{B}$ satisfies Properties 1) to 3) with probability at least
$1-Ne^{-\beta^2\alpha n/2}-N^2e^{-\beta^2\alpha^2 n/3}-N^2e^{-\beta^2\alpha(1-\alpha) n/2}$. Thus, as long as $N<e^{Rn}$, our procedure outputs a set $\mathfrak{B}$ satisfying Properties 1) through 3) with high probability. Further, we may let $N=2^{Rn}$, thus also satisfying Property 4). 

Next, given $\cB,\cB'\subseteq [n]$ such that $|\cB'\setminus\cB|>0$, for every $y^n\in\cY^n$ and $y\in\cY$, we have 
\begin{align} 
\hat{P}_{y^n}(y|\cB')&= \frac{\{i\in\cB':y_i=y\}}{|\cB'|}	\\
&= \frac{|\{i\in\cB'\setminus\cB:y_i=y\}|+|\{i\in\cB\cap\cB':y_i=y\}|}{|\cB'|}\\
&= \frac{|\cB'\setminus\cB|}{|\cB'|}\hat{P}_{y^n}(y|\cB'\setminus\cB)+\frac{|\cB\cap\cB'|}{|\cB'|}\hat{P}_{y^n}(y|\cB'\cap\cB).
\end{align}
Thus, by the triangle inequality for variational distance, we have
\begin{align}
	&\mathbb{V}\left(\hat{P}_{y^n}(\cdot|\cB'),\wy(\cdot|x',\szero)\right)\\
	&\geq\frac{|\cB'\setminus\cB|}{|\cB'|}\mathbb{V}\left(\hat{P}_{y^n}(\cdot|\cB'\setminus\cB),\wy(\cdot|x',\szero)\right) \\
	&\qquad\qquad- \frac{|\cB\cap\cB'|}{|\cB'|}\mathbb{V}\left(\hat{P}_{y^n}(\cdot|\cB'\cap\cB),\wy(\cdot|x',\szero)\right)\\
	& \geq  \frac{|\cB\setminus\cB'|}{|\cB'|}\mathbb{V}\left(\hat{P}_{y^n}(\cdot|\cB'\setminus\cB),\wy(\cdot|x',\szero)\right)-2 \frac{|\cB'\cap\cB|}{|\cB'|} .
\end{align}

Property 5) now follows from applying properties 1) to 3) to $\cB$ and $\cB'$.
\end{proof}

\end{document}